\documentclass[11pt]{paper}

\usepackage{fullpage}
\usepackage{cite}

\usepackage[utf8x]{inputenc}
\usepackage[T1]{fontenc}
\usepackage{lmodern}
\usepackage{microtype}
\usepackage[english]{babel}
\usepackage{graphicx}
\graphicspath{{img/}}
\usepackage{hyperref,color}
\usepackage[usenames,dvipsnames]{xcolor}
\usepackage{etoolbox}
\usepackage{setspace}
\usepackage{subfig}
\usepackage{enumitem}
\usepackage{tabularx}

\usepackage{amsmath,amssymb,amsfonts}
\usepackage{amsthm}
\usepackage{fge}
\usepackage{mathrsfs}
\theoremstyle{plain}
\newtheorem{theorem}{Theorem}[section]
\newtheorem{lemma}[theorem]{Lemma}

\newtheorem{corollary}[theorem]{Corollary}

\theoremstyle{definition}

\theoremstyle{remark}

\usepackage{xspace}

\newcommand{\abbrev}[2]{\expandafter\newcommand\csname #1\endcsname{#2\xspace}}
\abbrev{Caratheodory}{Carath\'eodory}
\abbrev{Barany}{B\'ar\'any}
\abbrev{Matousek}{Matou{\v{s}}ek}
\abbrev{Lovasz}{Lov{\'a}sz}
\abbrev{Frederic}{Fr\'ed\'eric}

\abbrev{suchthat}{s.t.\@}
\abbrev{wrt}{w.r.t.\@}
\abbrev{ie}{i.e.\@}
\abbrev{etal}{et al.\@}

\DeclareFontFamily{U}{mathx}{\hyphenchar\font45}
\DeclareFontShape{U}{mathx}{m}{n}{
      <5> <6> <7> <8> <9> <10>
      <10.95> <12> <14.4> <17.28> <20.74> <24.88>
      mathx10
      }{}
\DeclareSymbolFont{mathx}{U}{mathx}{m}{n}
\DeclareFontSubstitution{U}{mathx}{m}{n}
\DeclareMathAccent{\widecheck}{0}{mathx}{"71}
\DeclareMathAccent{\wideparen}{0}{mathx}{"75}

\newcommand{\eps}{\ensuremath{\varepsilon}}

\newcommand{\N}{\ensuremath{\mathbb{N}}}

\newcommand{\R}{\ensuremath{\mathbb{R}}}

\newcommand{\C}{\mathcal{C}}
\newcommand{\HsSet}{\mathcal{H}}
\newcommand{\Hs}{H}
\newcommand{\hp}{h}

\newcommand{\Test}{\ensuremath{\textsc{Test}}}

\newcommand{\dualshell}[1]{\ensuremath{\Pi({#1})}}

\newcommand{\ch}[1]{\ensuremath{\textsc{ch}(#1)}}
\newcommand{\fh}[1]{\ensuremath{\textsc{fh}(#1)}}

\newcommand{\polar}[1]{\ensuremath{{#1}^*}}

\newcommand{\zero}{\ensuremath{\mathbf{o}}}

\newcommand{\polarzero}[1]{\ensuremath{{#1}^*}}
\newcommand{\polarinf}[1]{\ensuremath{{#1}^*}}

\newcommand{\planezero}[1]{\ensuremath{#1_{_\zero}}}
\newcommand{\planeinf}[1]{\ensuremath{#1_{_\infty}}}

\newcommand{\planezeropolar}[1]{\ensuremath{\rho_{_\zero}(#1)}}
\newcommand{\planeinfpolar}[1]{\ensuremath{\rho_{_\infty}(#1)}}

\title{Asymmetric Convex Intersection Testing}

\author{%
  Luis Barba\thanks{%
    Department of Computer Science, ETH Z\"urich, Switzerland, \texttt{luis.barba@inf.ethz.ch}
    }
\and
  Wolfgang Mulzer\thanks{%
    Institut f\"ur Informatik, Freie Universit\"at Berlin,
    \texttt{mulzer@inf.fu-berlin.de}. Supported in part by ERC 
    StG 757609.
    }
}

\begin{document}
\maketitle

\abstract{We consider \emph{asymmetric convex intersection 
testing} (ACIT).  

Let $P \subset \mathbb{R}^d$ be a set of $n$ points 
and $\mathcal{H}$ a set of $n$ halfspaces in $d$ dimensions. 
We denote by $\ch{P}$ the polytope obtained by taking the 
convex hull of $P$, and by $\fh{\mathcal{H}}$ 
the polytope obtained by 
taking the intersection of the halfspaces in $\mathcal{H}$. 
Our goal is to decide whether the 
intersection of $\mathcal{H}$ and the convex hull of $P$ are disjoint. 
Even though ACIT is a natural variant of classic LP-type problems
that have been studied at length in the literature, and despite
its applications in the analysis of high-dimensional data sets, it
appears that the problem has not been studied before. 

We discuss how known approaches can be used to attack the 
ACIT problem, and we provide a very simple strategy that leads
to a deterministic algorithm, linear on $n$ and $m$, whose running time
depends reasonably on the dimension $d$.
}

\section{Introduction}

Let $d \in \N$ be a fixed constant.
Convex polytopes in dimension $d$ can be implicitly represented in two ways, either by its set of vertices, 
or by the set of halfspaces whose intersection defines the polytope. 
A polytope represented by its vertices is usually called a \emph{V-polytope}, while a polytope represented by a set of halfspaces is known as an \emph{H-polytope}. 
Note that the actual complexity of the polytopes can be much larger than the size of their representations~\cite[Theorem~5.4.5]{Matousek02}. 
In this paper, we study the problem of testing the intersection of convex polytopes with different implicit representations.
When both polytopes have the same representation, testing for their intersection reduces to linear programming. 
However, when there is a mismatch in the representation, the problem changes in nature and becomes more challenging.

To formalize our problem, let $P \subset \R^d$ be a set of $n$ points in $\R^d$, and let 
$\HsSet$ be a set of $n$ halfspaces in $\R^d$.\footnote{We will 
assume that both $P$ and $\HsSet$ are in
\emph{general position} (the exact meaning of this will
be made clear later)}
Just as $P$  implicitly defines the polytope $\ch{P}$ obtained by taking the convex hull of $P$, the set $\HsSet$ implicitly defines  
the polytope $\fh{\HsSet}$ obtained by taking the intersection
of the halfspaces in $\HsSet$. In the \emph{asymmetric convex 
intersection problem} (ACIT), our goal is to decide whether the 
intersection of $\HsSet$ and the convex hull of $P$ are disjoint. 

We may assume that $\fh{\HsSet}$ is nonempty. 
Otherwise, ACIT becomes trivial. If $\ch{P}$ and $\fh{\HsSet}$
intersect, we would like to find a \emph{witness point} in 
both $\ch{P}$ and $\fh{\HsSet}$; if not, we would 
like to determine the closest pair between $\ch{P}$
and $\fh{\HsSet}$ and a separating hyperplane.

Even though ACIT seems to be a natural problem that fits well
into the existing work on algorithmic aspects of high-dimensional
polytopes~\cite{BarbaLa15}, we are not aware of any prior work on it. 
While intersection detection of convex polytopes has been a central topic in computational geometry~\cite{edelsbrunner1985computing,muller1978finding,dobkin1985linear,Chazelle92anoptimal,Chazelle1987,ChazelleD80},
when we deal with an intersection test between a V-polytope and an H-polytope, the problem seems to remain unstudied. 
Even the seemingly easy case of this problem in dimension $d=2$ has no trivial solution running in linear time.

The lack of a solution for ACIT may be even more surprising considering that ACIT can be used in the analysis
of high-dimensional data: given a high-dimensional data set, represented
as a point cloud $P$, it is natural to represent the \emph{interpolation}
of the data as the convex hull $\ch{P}$. Then, we would like to
know whether the interpolated data set contains an item that satisfies
certain \emph{properties}. These properties are usually represented as 
linear constraints that must be satisfied, i.e., the data point must belong to the intersection of a set of halfspaces. 
Then, a witness point corresponds to an interpolated data point with the desired properties,
and a separating hyperplane may indicate which properties cannot be fulfilled
by the data at hand.

Even though ACIT has not been addressed before, several 
approaches for related problems\footnote{In particular, checking 
for the intersection of the convex hulls of to $d$-dimensional point sets}
may be used to attack the problem.
The range of techniques goes from simple brute-force, over 
classic linear programming~\cite{Chvatal83}, the theory of LP-type 
problems~\cite{Chazelle01,SharirWe92} (also in 
implicit form~\cite{Chan04}), to parametric search~\cite{Matousek93}. 
In Section~\ref{sec:simple},
we will examine these in more detail and discuss their merits and
drawbacks. Briefly, several of these approaches can be applied to 
ACIT. However, as we will see, it seems hard to get an algorithm
that is genuinely simple and at the same time achieves linear (or almost
linear) running time in the number of points and halfspaces, with a 
reasonable dependency on the dimension $d$.

Thus, in Section~\ref{sec:algorithm},
we present a simple recursive primal-dual pruning strategy that
leads to a deterministic linear time algorithm with a dependence on $d$
that is comparable to the best bounds for linear programming. Even though
the algorithm itself is simple and can be presented in a few lines, the
analysis requires us to take a close look at the polarity transformation
and how it interacts with two disjoint polytopes (Section~\ref{sec:tools}).
Its analysis is also non-trivial and its correctness spans over the entire Section~\ref{sec:Correctnes}.
We believe in the development of simple and efficient methods. The analysis can be complicated, 
but the algorithm must remain simple. 
The simpler the algorithm, the more likely it is to be eventually implemented.

\section{How to solve ACIT with existing tools}\label{sec:simple}

The first thing that might come to mind to solve ACIT is to cast it as a linear program.
This is indeed possible, however the resulting linear program consists of $n$ variables and 
$O(n + m)$ constraints. We want to find a point $x$ 
subject to being inside all halfspaces in $\HsSet$, and being a 
convex combination of all points in $P$. That is, we want 
$x = \sum_{p \in P} \alpha_p\, p$, where 
$\sum_{p \in P} \alpha_p = 1$, and $\alpha_p \geq 0$, for all 
$p \in P$. Moreover, we want that $x \in \Hs$, for all 
$\Hs \in \HsSet$, which can be expressed as $m$ linear inequalities
by looking at the scalar product of $x$ and the normal vectors of the
bounding hyperplanes of the halfspaces.  Because the best combinatorial algorithms
for linear programming provide poor running times when both the 
number of variables and constraints are large, this approach is far
from efficient unless $n$ is really small.\footnote{In fact, in the
traditional computational model of computational geometry, the 
\textsc{Real RAM}~\cite{PreparataSh85}, we cannot solve general 
linear programs in polynomial time, since the best known algorithms
(e.g., ellipsoid, interior point methods)
are only \emph{weakly} polynomial with a running time that depends on
the bit complexity of the input.} 

Another trivial way to solve ACIT, the \emph{brute force} algorithm,
is to compute all facets of $\ch{P}$. That is, we can compute 
$\ch{P}$ explicitly to obtain a set $\HsSet_P$ of the 
$O(n^{\lfloor d/2\rfloor})$ halfspaces with 
$\ch{P} = \fh{\HsSet_P}$~\cite{ClarksonSh89,Chazelle93}. With this 
representation, we can test if $\fh{\HsSet_P}$ and $\fh{\HsSet}$ 
intersect using a general linear program with $d$ variables, or 
compute the distance between $\fh{\HsSet_P}$ and $\fh{\HsSet}$ using
either an LP-type algorithm (see below), or algorithms for convex
quadratic programming~\cite{KapoorVa86,KozlovTaHa80}. The running
time is again quite bad for larger values of $n$ and $d$, since the
size of $\HsSet_P$ might be as high as $\Theta(n^{\lfloor d/2\rfloor})$~\cite[Theorem~5.4.5]{Matousek02}. 

A more clever approach is to use the LP-type framework directly, as described below.

\subsubsection*{The LP-type Framework}

The classic \emph{LP-type framework} that was introduced
by Sharir and Welzl~\cite{SharirWe92} in order to 
extend the notion of low-dimensional linear programming
to a wider range of problems.  An LP-type problem
$(\C, w)$ consists of a set $\C$ of $k$ \emph{constraints} and
a \emph{weight function} $w : 2^{\C} \rightarrow \R$ that assigns
a real-valued weight $w(C)$ to each set $C \subseteq \C$ of 
constraints.\footnote{Actually, we can allow weights from
an arbitrary totally ordered set, but for our purposes, real 
weights will suffice.}
The weight function must satisfy the following three axioms:
\begin{itemize}
  \item \textbf{Monotonicity}: For any set $C \subseteq \C$ of
    constraints and any $c \in \C$, we have
    $w\big(C \cup \{ c \}\big) \leq w(C)$.
  \item \textbf{Existence of a Basis}: There is a 
    constant $\tilde{d} \in \N$ such that for any $C \subseteq \C$,
    there is a subset $B \subseteq C$ with $|B| \leq \tilde{d}$ and
    $w(B) = w(C)$.
  \item \textbf{Locality}: For any $B \subseteq C \subseteq \C$
    with $w(B) = w(C)$ and for any $c \in C$, we have
    that if $w\big(C \cup \{c\}\big) <  w(C)$, then also 
    $w\big(B \cup \{c\}\big) < w(B)$.
\end{itemize}
For $C \subseteq \C$, an inclusion-minimal subset $B \subseteq C$
with $w(B) = w(C)$ is called a \emph{basis} for $C$. 
Solving an LP-type problem $(\C, w)$ amounts to computing 
a basis for $\C$. Many algorithms have been developed for this extension of
linear programming, provided that base cases with a constant number of constraints can be solved in $O(1)$ time.
Seidel proposed a simple randomized algorithm with expected $O(\tilde{d}! k)$ running time~\cite{Seidel91}.
From there, several algorithms have been introduced improving the dependency on $\tilde{d}$ in the running time~\cite{chan2018improved,Clarkson88,Seidel91,SharirWe92}.
The best known randomized algorithm solves LP-type problems in $O(\tilde{d}^2 k + 2^{O(\sqrt{\tilde{d} \log \tilde{d}})})$ time, 
while the best deterministic algorithms have still a running time of the form $O(\tilde{d}^{O(\tilde{d})} k)$.
We would like to obtain an algorithm with a similar running time for ACIT.

\subsubsection*{ACIT as an LP-type problem}

To use these existing machinery, one can try to cast ACIT as an LP-type problem.
To this end, we fix $\HsSet$, and define an LP-type 
problem $(P, w)$ as follows. The \emph{constraints} are modeled 
by the points in $P$. The weight function
$w: 2^P \to \R$ is defined as $w(Q) = d\big(\ch{Q}, \fh{\HsSet}\big)$, 
for any $Q \subseteq P$, where $d(\cdot, \cdot)$ is the 
smallest Euclidean distance between any pair of points from the
two polytopes. It is a pleasant exercise to show that this indeed 
defines an LP-type problem of combinatorial dimension $d$.

Thus, the elegant methods to solve LP-type problems mentioned above become applicable.  
Unfortunately, this does not give an efficient algorithm. This is because 
the set $\HsSet$ remains fixed throughout, making unfeasible 
to solve the base cases consisting of $O(1)$ constraints of $P$ in constant time. 

\subsubsection*{A randomized algorithm for ACIT}
As an extension of the LP-type framework, Chan~\cite{Chan04} introduced a new technique that allows us to deal with certain LP-type problems where the constraints are too numerous to write explicitly, and are instead specified ``implicitly''. 
More precisely, as mentioned above, ACIT can be seen as a linear program, with $m$ constraints coming from $\HsSet$, and $O(n^{\lfloor d/2\rfloor})$ constraints coming from all the facets of $\ch{P}$. The latter are implicitly defined by $P$ using only $n$ points. 
Thus an algorithm capable of solving implicitly defined linear programs would provide a solution for ACIT.
The technique developed by Chan achieves this by using two main ingredients: a decision algorithm, and a partition of the problem into subproblems of smaller size whose recursive solution can be combined to produce the global solution of the problem.
Using the power of randomization and geometric cuttings, this technique leads to a complicated algorithm to solve this implicit linear program, and hence ACIT, in expected $O(d^{O(d)} (n + m))$ time~\cite{ChanPersonal}. 
Besides the complexity of this algorithm, the constant hidden by the big $O$ notation resulting from using this technique seems prohibitive~\cite{louchard2007randomized}. 

In hope of obtaining a deterministic algorithm for ACIT, one can turn to
multidimensional parametric search~\cite{Matousek93} to try de-randomizing the above algorithm.
However, even if all the requirements of this technique can be sorted out, it would lead to a highly complicated algorithm and
polylogarithmic overhead.

In the following sections, we present the first deterministic solution for ACIT using a simple algorithm that overcomes
the difficulties mentioned above. We achieve this solution by diving into the 
intrinsic duality of the problem provided by the polar transformation, while exploiting the LP-type-like structure of our problem.
The resulting algorithm is quite simple, and a randomized version of it could be written with a few lines of code,
provided that one has some LP solver at hand.

\section{Geometric Preliminaries}\label{sec:tools}

Let $\zero$ denote the \emph{origin} of $\R^d$.
A \emph{hyperplane} $\hp$ is a $(d - 1)$-dimensional affine 
space in $\R^d$ of the form
\[
  \hp = \{x \in \R^d \mid \langle z, x\rangle = 1\},
\]
for some $z \in \R^d \setminus \{ \zero \}$, 
where $\langle \cdot, \cdot \rangle$ 
represents the scalar product in $\R^d$. We exclude hyperplanes that 
pass through the origin. A (closed) \emph{halfspace} is the closure 
of the point set on either side of a given hyperplane, i.e., a
halfspace contains the hyperplane defining its boundary.

\subsection{The Polar Transformation}
Given a point $p \in \R^d$, we define the \emph{polar} of $p$ 
to be the hyperplane 
\[
  \polar{p}  = \{ x \in \R^d \mid \langle p, x \rangle = 1\}. 
\]
Given a  hyperplane $\hp$ in $\R^d$, we define its 
\emph{polar} $\polar{\hp} \in \R^d$ as the point with
\[
\hp = \{x \in \R^d \mid \langle x, \polar{\hp}\rangle = 1\}.
\]
Let $\planezeropolar{p} = \{x \in \R^d \mid \langle p, x\rangle \leq 1\}$ 
and $\planeinfpolar{p} = \{x \in \R^d \mid \langle p, x\rangle \geq 1\}$ be 
the two halfspaces supported by $\polar{p}$ such that
$\zero \in \planezeropolar{p}$ and $\zero \notin \planeinfpolar{p}$.
Similarly, $\planezero{\hp}$ and $\planeinf{\hp}$ denote the halfspaces 
supported by $\hp$ such that $\zero \in \planezero{\hp}$ and 
$\zero \notin \planeinf{\hp}$.

Note that the polar of a point $p \in \R^d$ is a hyperplane 
whose polar is equal to $p$, i.e., the polar operation is involutory 
(for more details, see Section~2.3 in Ziegler's book~\cite{Ziegler95}).
The following result is illustrated in 
Figure~\ref{fig:PolarExamples}(a), for $d = 2$.

\begin{lemma}[Lemma 2.1 of~\cite{BarbaLa15}]\label{lem:polarity}
Let $p$ and $\hp$ be a point and a hyperplane in $\R^d$, 
respectively. Then, $p \in \planezero{\hp}$ if and only if 
$\polar{\hp} \in \planezeropolar{p}$. Also, $p \in \planeinf{\hp}$ 
if and only if $\polar{\hp} \in \planeinfpolar{p}$. 
Finally, $p \in \hp$ if and only if $\polar{\hp} \in \polar{p}$.
\end{lemma}

\begin{figure}[t]
\centering
\includegraphics{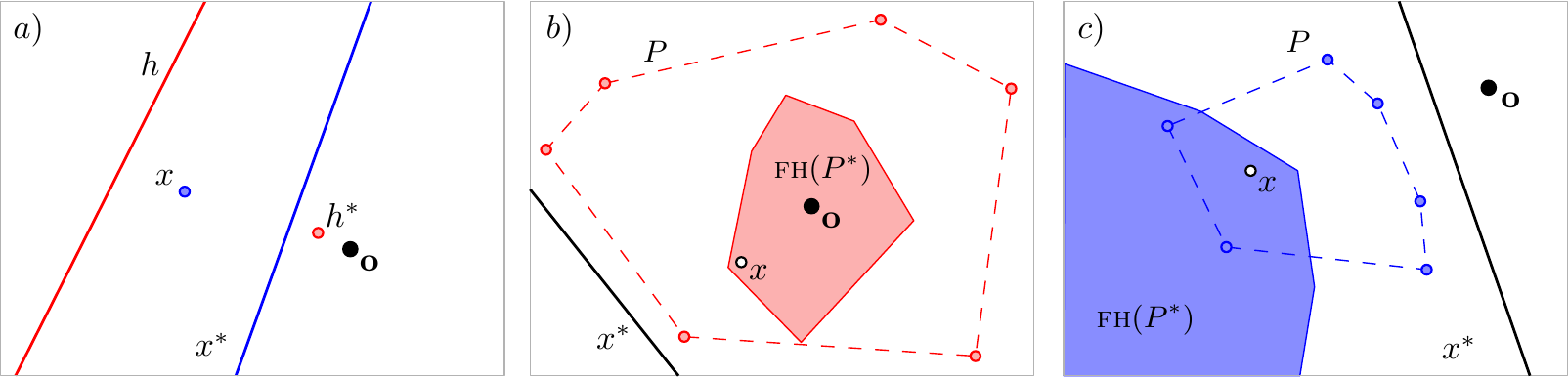}
\caption{\small (a) The situation described in 
Lemma~\ref{lem:polarity}. (b) A valid set $P$ of points that is embracing and its polar $\polar{P}$ that is also embracing. 
(c) A set $P$ that is avoiding and its
polar $\polar{P}$ that is also avoiding.}
\label{fig:PolarExamples}
\end{figure}

Let $P$ be a set of points in $\mathbb{R}^d$.
We say that $P$ is \emph{embracing} if $\zero$ lies in the interior 
of $\ch{P}$. We say that $P$ is \emph{avoiding} if $\zero$ lies in the 
complement of $\ch{P}$. Note that we do not consider point sets 
whose convex hull has $\zero$ on its boundary.
We say that $P$ is \emph{valid} if it is either embracing or avoiding. 

Let $\HsSet$ be a set of halfspaces in $\R^d$ such that 
$\fh{\HsSet}\neq \emptyset$, and the boundary of no halfspace in 
$\HsSet$ contains $\zero$.
We say that $\HsSet$ is \emph{embracing} if  $\zero \in \Hs$ for all 
$\Hs \in \HsSet$ (i.e., $\zero \in \fh{\HsSet}$).
We say that $\HsSet$ is \emph{avoiding} if none of its 
halfspaces contains $\zero$, i.e., 
$\zero \notin \bigcup_{\Hs \in \HsSet} H$.
We say that $\HsSet$ is \emph{valid} if it is either embracing or avoiding.

We now describe how to polarize convex polytopes defined as 
convex hulls of valid sets of points or as intersections of 
valid sets of halfspaces.
Let $\HsSet$ be a valid set of halfspaces in $\R^d$.
To \emph{polarize} $\HsSet$, consider the set of hyperplanes bounding the halfspaces in $\HsSet$, and
let $\polar{\HsSet}$ be the set consisting of all the points being the polars of these hyperplanes. 

\begin{lemma}\label{lemma: zero inside for valid}
Let $\HsSet$ be a valid set of halfspaces in $\R^d$. Then, 
$\polar{\HsSet}$ is embracing if and only if $\HsSet$ is embracing.
\end{lemma}
\begin{proof}
Recall that $\HsSet$ is embracing if and only if $\fh{\HsSet}$ is bounded 
and contains $\zero$.

$\Rightarrow)$. Assume that $\polar{\HsSet}$ is embracing. Thus,  
$\zero \in \ch{\polar{\HsSet}}$. 
In this case, there is a subset $Q$ of $d+1$ points of 
$\polar{\HsSet}$ whose convex hull contains $\zero$, by 
Carath\'eodory's theorem~\cite[Theorem~1.2.3]{Matousek02}. 
Consider all halfspaces of $\HsSet$ whose boundary polarizes to a point in 
$Q$. If none of these halfspaces contains the origin, then their 
intersection has to be empty. This is not allowed by the validity of 
$\HsSet$. Thus, as $\HsSet$ is valid, and as $\HsSet$ cannot avoid the 
origin, we conclude that $\HsSet$ is embracing.

$\Leftarrow)$. For the other direction, assume that 
$\zero \notin \ch{\polar{\HsSet}}$. We want to prove that 
$\HsSet$ is not embracing. For this, let $\hp$ be a hyperplane 
that separates $\zero$ from $\ch{\polar{\HsSet}}$. 
That is, $\polar{\HsSet} \subset \planeinf{\hp}$. 
Lemma~\ref{lem:polarity} implies that the segment $\zero \polar{\hp}$ 
intersects the boundary of each plane in $\HsSet$. Therefore, since the 
ray shooting from $\zero$ in the direction of the vector 
$-\polar{\hp}$ intersects no plane bounding a halfspace in $\HsSet$, 
the polytope $\fh{\HsSet}$ either does not contain the origin or 
is not bounded.  Consequently, $\HsSet$ is not embracing.
\end{proof}

Let $P$ be a valid set of points in $\R^d$.
To \emph{polarize} $P$, let $\dualshell{P}$ be 
the set of hyperplanes polar to the points of $P$.
We have two natural ways of polarizing $P$, 
depending on whether $\zero$ lies in the interior of $\ch{P}$, or in its complement (recall that $\zero$ cannot lie on the boundary of $\ch{P}$). 
If~$\zero \in \ch{P}$, then 
\[
  \polar{P} = \big\{\planezero{\hp} \mid \hp \in \Pi(P)\big\}
\]
is the 
\emph{polarization} of $P$.  Otherwise, if $\zero \notin \ch{P}$, 
then 
\[
  \polar{P} = \big\{\planeinf{\hp} \mid \hp \in \Pi(P)\big\}.  
\]

\begin{lemma}\label{lemma:Valid leads to valid}
Let $P$ be a valid set of points in $\R^d$. Then $\polar{P}$ is valid, 
i.e., $\fh{\polar{P}} \neq \emptyset$ and $\polar{P}$ is either 
embracing or avoiding.
\end{lemma}

\begin{proof}
If $\zero \notin \ch{P}$, then there is a hyperplane $\hp$ such that 
$P\subset \planeinf{\hp}$. Thus, $\polar{\hp}$ belongs to 
$\polarinf{p}$ for every $p\in P$, i.e., $\polar{\hp}\in \fh{\polar{P}}$. 
Thus, $\fh{\polar{P}}$ is nonempty, and none of its halfspaces contains 
the origin by definition. That is, $\polar{P}$ is avoiding.
If $\zero \in \ch{P}$, then $\zero\in \fh{\polar{P}}$ by definition. Thus, 
to show that $\polar{P}$ is embracing, it remains only to show that it is 
bounded. To this end, assume for a contradiction that $\fh{\polar{P}}$ is 
unbounded.  Then, we can take a point $x \in \fh{\polar{P}}$ at arbitrarily 
large distance from $\zero$. 
Thus, $\polar{x}$ is a plane arbitrarily close to $\zero$ such that 
$P \subset \polarzero{x}$.
Therefore, all points of $P$ must lie on a single halfspace that 
contains $\zero$ on its boundary. 
Because $P$ is valid, we know that $\zero$ cannot lie on the boundary of 
$\ch{P}$ and hence, $\zero\notin \ch{P}$---a contradiction with our 
assumption that $\zero \in \ch{P}$. Therefore, $\fh{\polar{P}}$ is 
bounded and hence $\polar{P}$ is embracing. 
\end{proof}

\begin{lemma}\label{lemma: Involutory operation}
Let $P \subset \R^d$ be a valid finite point set in $d$ dimensions, and let 
$\HsSet$ be a valid finite set of halfspaces in $d$ dimensions.
Then the polar operator is involutory:
$P = \polar{(\polar{P})}$ and $\HsSet = \polar{(\polar{\HsSet})}$.
\end{lemma}

\begin{proof}
The equality $P = \polar{(\polar{P})}$ follows directly from the 
definition, because the polar operator for points and hyperplanes 
is involutory. For equality $\HsSet = \polar{(\polar{\HsSet})}$, we must 
check that the orientation of the halfspaces is preserved. 
First,  if $\HsSet$ is embracing, i..e, $\zero \in \fh{\HsSet}$, then
every $\Hs \in \HsSet$ is of the form $\Hs = \planezero{\hp}$, for some 
$(d - 1)$-dimensional hyperplane $\hp$. Moreover, Lemma~\ref{lemma: zero inside for valid} implies that 
$\zero \in \ch{\polar{\HsSet}}$. Thus, we have 
$\HsSet = \polar{(\polar{\HsSet})}$ in this case.
Similarly, if $\HsSet$ is avoiding, i.e.,
$\zero \not\in \bigcup_{\Hs \in \HsSet} \Hs$, then 
every $\Hs \in \HsSet$ is of the form $H = \planeinf{\hp}$ for some 
$(d - 1)$-dimensional
hyperplane $\hp$, and by Lemma~\ref{lemma: zero inside for valid},
we have $\zero \not\in \ch{\polar{\HsSet}}$.
Thus, we have again $\HsSet = \polar{(\polar{\HsSet})}$.
\end{proof}

\begin{corollary}\label{corollary: duality}
Let $P$ be a valid set of points in $\mathbb{R}^d$. Then, $\polar{P}$ is embracing if and only if $P$ is embracing. 
Moreover, $\polar{P}$ is avoiding if and only if $P$ is avoiding.
\end{corollary}
\begin{proof}
Because $P$ is valid, $\polar{P}$ is valid by Lemma~\ref{lemma:Valid leads to valid}. 
Therefore, Lemma~\ref{lemma: zero inside for valid} implies that $\polar{P}$ is embracing if and only if $\polar{(\polar{P})}$ is embracing.
Because $P = \polar{(\polar{P})}$ by Lemma~\ref{lemma: Involutory operation}, we conclude that $\polar{P}$ is embracing if and only if $P$ is embracing.
Note that if a valid set $P$ is not embracing, then it is avoiding, yielding the second part of the result.
\end{proof}

\begin{figure}[h]
\centering
\includegraphics{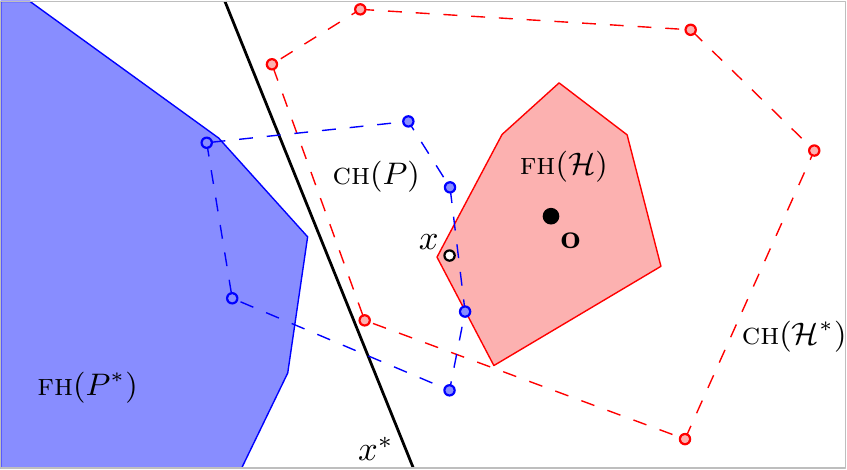}
\caption{\small An example of 
Theorem~\ref{thm:polarPolyhedra} in dimension 2, where a point $x$ lies 
in the intersection of $\ch{P}$ and $\fh{\HsSet}$ if and only if $\polar{x}$ 
separates $\ch{\polar{\HsSet}}$ from $\fh{\polar{P}}$.}
\label{fig:PolarityOfConvexPolyhedra}
\end{figure}

The following result is illustrated in 
Figure~\ref{fig:PolarityOfConvexPolyhedra}, for $d = 2$.

\begin{theorem}[Consequence of 
Theorem 3.1 of~\cite{BarbaLa15}]\label{thm:polarPolyhedra}
Let $P$ be a finite set of points and let $\HsSet$ be a valid finite set 
of halfspaces in $\R^d$ such that either 
\textup{(1)} $P$ is avoiding while $\HsSet$ is embracing, or 
\textup{(2)}~$P$~is embracing while $\HsSet$ is avoiding. 
Then, a point $x$ lies 
in the intersection of $\ch{P}$ and $\fh{\HsSet}$ if and only if 
the hyperplane $\polar{x}$ separates $\fh{\polar{P}}$ from 
$\ch{\polarzero{\HsSet}}$.  Also a hyperplane $\hp$ separates 
$\ch{P}$ from $\fh{\HsSet}$ if and only if the point $\polar{\hp}$ 
lies in the intersection of $\fh{\polar{P}}$ and $\ch{\polar{\HsSet}}$.
\end{theorem}

Conditions~(1) and (2) will be crucial in our algorithm. 
Note that by Corollary~\ref{corollary: duality}, we have that $P$ and $\HsSet$ satisfy condition~(1), then the point set 
$\polar{\HsSet}$ and the set $\polar{P}$ of halfspaces satisfy 
condition~(2), and vice versa.

\subsection{Conflict Sets, $\eps$-nets, and Closest Pairs}

Let $P \subseteq \R^d$ be a finite point set in $d$ dimensions, 
and let $\Hs$ be a halfspace in $\R^d$. We say that a point 
$p \in P$ \emph{conflicts} with $\Hs$ if $p \in \Hs$. The 
\emph{conflict set} of $P$ and $\Hs$, denoted $V_\Hs(P)$, 
consists of all points $p \in P$ that are in conflict with $\Hs$,
i.e., $V_\Hs(P) = P \cap \Hs$.
Let $\eps \in (0,1)$ be a parameter. A set $N \subseteq P$ is called 
an \emph{$\eps$-net} for $P$ if for every halfspace $\Hs$ in 
$\R^d$, we have \begin{equation}\label{equ:eps_net}
V_H(N) = \emptyset \Rightarrow \big|V_H(P)\big| < \eps |P|.
\end{equation}
By the classic 
\emph{$\eps$-net theorem}~\cite[Theorem~5.28]{HarPeled11}, a random
subset $N \subset P$ of size 
$\Theta\Big(\eps^{-1} \log\big(\eps^{-1} + \alpha^{-1}\big)\Big)$ 
is an $\eps$-net for $P$ with probability at least $1 - \alpha$.
For a deterministic algorithm running in linear time, we can 
compute such a net using the complicated algorithm of Chazelle and 
Matou\v{s}ek~\cite[Chapter~4.3]{Chazelle01} or the much simpler 
algorithm introduced by Chan~\cite{chan2018improved}.
See the textbooks of Matou\v{s}ek~\cite{Matousek02},
Chazelle~\cite{Chazelle01}, or Har-Peled~\cite{HarPeled11} for more 
details on $\eps$-nets and their uses in computational geometry.
The following observation shows the usefulness of conflict sets 
for our problem.

\begin{lemma}\label{lem:cp_violator_primal}
Let $P \subseteq \R^d$ be a finite point set 
and $\HsSet$ a finite set of halfspaces in $d$ dimensions.
Let $N \subseteq P$ such that $\fh{\HsSet}$ and $\ch{N}$ are disjoint, 
and let $x, y$ be the closest pair between them, such that
$x \in \fh{\HsSet}$ and $y \in \ch{N}$. Let $\Hs_y$ be the halfspace 
through $y$ perpendicular to the segment $xy$, containing 
$\fh{\HsSet}$. Then, we have
$d\big(\fh{\HsSet}, P\big) < d\big(\fh{\HsSet}, N\big)$ if and only if 
$V_{\Hs_y}(P) \neq \emptyset$.
\end{lemma}

\begin{proof}
Since all points in $\R^d \setminus \Hs_y$ have distance larger than 
$d\big(\fh{\HsSet}, N\big)$ from $\fh{\HsSet}$, the implication
$V_{\Hs_y}(P) = \emptyset \Rightarrow 
d\big(\fh{\HsSet}, P\big) < d\big(\fh{\HsSet}, N\big)$ is immediate.

Now assume that $V_{\Hs_y}(P) \neq \emptyset$, say, 
$p \in V_{\Hs_y}(P)$. Then, the line segment $py$ is contained 
in $\ch{P}$, and since $p \in V_{\Hs_y}(P)$ and since $p$ does not
lie on the boundary of $\Hs_y$ be our general position assumption, 
it follows that the 
angle between the segments $py$ and $xy$ is strictly smaller than
$\pi/2$.
Hence, we have 
\[
  d\big(\fh{\HsSet}, P\big) \leq d(x, py) < \big(\fh{\HsSet}, N\big),
\]
as claimed.
\end{proof}

Similarly, let $\HsSet$ be a finite set of halfspaces in 
$d$ dimensions, and let $p \in \R^d$ be a point. The \emph{conflict 
set} $V_p(\HsSet)$ of $\HsSet$ and $p$ consists of all halfspaces 
that do not contain $p$, i.e. $V_p(\HsSet) = \{ \Hs \in \HsSet \mid 
p \not\in \Hs \}$. We have the following polar version of 
Lemma~\ref{lem:cp_violator_primal}:

\begin{lemma}\label{lem:cp_violator_polar}
Let $P \subseteq \R^d$ be a finite point set 
and $\HsSet$ a finite set of halfspaces in $d$ dimensions.
Let $\HsSet' \subseteq \HsSet$ such that 
$\fh{\HsSet'}$ and $\ch{P}$ are disjoint, 
and let $x, y$ be the closest pair between them, such that
$x \in \fh{\HsSet'}$ and $y \in \ch{P}$. 
Then, we have
$d\big(\fh{\HsSet}, P\big) > d\big(\fh{\HsSet'}, P\big)$ if and only if 
$V_{x}(\HsSet) \neq \emptyset$.
\end{lemma}

\begin{proof}
First, if $V_x(\HsSet) = \emptyset$, then $x \in \fh{\HsSet}$, and 
since $\fh{\HsSet} \subseteq \fh{\HsSet'}$, we have 
$d\big(\fh{\HsSet}, P\big) = d\big(\fh{\HsSet'}, P\big)$.

Second, suppose that 
$V_x(\HsSet) \neq \emptyset$, say, $\Hs \in V_x(\HsSet)$. Then, 
$x \not\in \Hs$, and since, by gneral position, $x$ is the unique point 
in $\fh{\HsSet'}$ with $d(x, \ch{P}) = d(\fh{\HsSet'}, \ch{P})$,
we have
\[
  d\big(\fh{\HsSet}, P\big) \geq 
  d\big(\fh{\HsSet' \cup \{\Hs\}}, P\big)
  > d\big(\fh{\HsSet'}, P\big),
\]
as claimed.
\end{proof}

The following lemma gives a polar meaning to the notion of $\eps$-nets.

\begin{lemma}\label{lem:dualNet}
Let $N \subseteq P$ be an $\eps$-net for $P$ such that
if $\zero \in \ch{P}$, then also $\zero \in \ch{N}$.
For any point $x \in \R^d$, it 
holds that if $x \in \fh{\polar{N}}$, then 
$\big|V_x\big(\polar{P}\big)\big| \leq \eps|P|$.
\end{lemma}

\begin{proof}
First, suppose that $\zero \in \ch{P}$, then, we have $\zero \in \ch{N}$,
and hence $\zero \in \fh{\polar{N}}$. Since $x \in \fh{\polar{N}}$,
we have $x \in \planezeropolar{p}$, for all $p \in N$. 
By Lemma~\ref{lem:polarity},
we get $p \in \planezeropolar{x}$, for all $p \in N$, so 
$N \cap \planeinfpolar{x} = \emptyset$. Since $N$ is an $\eps$-net for $P$, 
we conclude $|P \cap \planeinfpolar{x}| \leq \eps |P|$. The claim now
follows, because by Lemma~\ref{lem:polarity}, we have
$|P \cap \planeinfpolar{x}| = |V_x\big(\polar{P})|$.

Second, suppose that $\zero \not\in \ch{P}$, then, we also
get $\zero \not\in \ch{N}$,
and hence $\zero \not\in \bigcup_{\Hs \in \polar{N}} \Hs$. 
Since $x \in \fh{\polar{N}}$,
we have $x \in \planeinfpolar{p}$, for all $p \in N$. 
By Lemma~\ref{lem:polarity},
we get $p \in \planeinfpolar{x}$, for all $p \in N$, so 
$N \cap \planezeropolar{x} = \emptyset$. Since $N$ is an $\eps$-net for $P$, 
we conclude $|P \cap \planezeropolar{x}| \leq \eps |P|$. The claim now
follows, because by Lemma~\ref{lem:polarity}, we have
$|P \cap \planezeropolar{x}| = |V_x\big(\polar{P})|$.
\end{proof}

\section{A Simple Algorithm}\label{sec:algorithm}

Let $P$ be a valid set of $n$ points and let $\HsSet$ be a valid set of 
$m$ halfspaces in $\R^d$ such that either 
\textup{(1)} $P$ is avoiding while $\HsSet$ is embracing, or 
\textup{(2)}~$P$~is embracing while $\HsSet$ is avoiding. 
We first present a 
slightly more restrictive 
algorithm that requires conditions (1) or (2) to hold. 
We spend the next few sections proving its correctness and 
running time, and then we extend it to a general algorithm for 
the ACIT problem.

\subsection{Description of the Algorithm}

Our algorithm $\Test(P, \HsSet)$ takes $P$ and $\HsSet$ as input,
such that either~(1) or~(2) is satisfied, and it computes either 
the closest pair between $\ch{P}$ and $\fh{\HsSet}$, if $\ch{P}$ and
$\fh{\HsSet}$ are disjoint, or the closest pair between $\ch{\polar{H}}$ and 
$\fh{\polar{P}}$, if $\ch{P}$ and $\fh{\HsSet}$ intersect. 
By Theorem~\ref{thm:polarPolyhedra}, this is always possible.

The algorithm is recursive.
Let $\alpha = c\, d^4 \log d$, for some
approriate constant $c > 0$.
For the base case, if both $|P|, |\HsSet| \leq \alpha$,
we apply the brute force algorithm: we explicity compute 
the polytope $\ch{P}$ to obtain the set $\HsSet_P$ of hyperplanes
with $\ch{P} = \fh{\HsSet_P}$, and we use a classic LP-type algorithm
to find the closest pair between $\ch{P}$ and $\fh{\HsSet}$ or
between $\fh{\polar{P}}$ and $\ch{\polar{\HsSet}}$.
Otherwise, we compute a $(1/d^4)$-net $N \subseteq P$, 
and if necessary, we add
$d + 1$ points to $N$ such that if $\zero \in \ch{P}$, then
$\zero \in \ch{N}$. These $d + 1$ points can be found in $O(n)$ time 
using basic linear algebra. Then, we execute the following 
loop.

\noindent
\textbf{Repeat $2d+1$ times:} Recursively call 
$\Test(\polar{\HsSet}, \polar{N})$; there are two possibilities.

\noindent\textbf{Case 1:} 
$\ch{\polar{\HsSet}}$ and $\fh{\polar{N}}$ are disjoint. Then,
$\Test(\polar{\HsSet}, \polar{N})$ returns the closest pair 
$x$, $y$, with $x \in \ch{\polar{\HsSet}}$ and $y \in \fh{\polar{N}}$ (unique by our general position assumption).
Let $V_y \subset \polar{P}$ be conflict set of $\polar{P}$ 
and $y$. If $V_y = \emptyset$, then 
report that $\ch{P}$ and $\fh{\HsSet}$ intersect, 
and output $x$, $y$ as the polar witness.
Otherwise, add to $N$ all elements of $\polar{V_y}$, 
and continue with the next iteration.

\noindent\textbf{Case 2:}
$\ch{\polar{\HsSet}}$ and 
$\fh{\polar{N}}$ intersect. Then,
$\Test(\polar{\HsSet}, \polar{N})$ returns the closest pair 
$x$, $y$ between $\fh{\polar{(\polar{\HsSet})}} = \fh{\HsSet}$ and 
$\ch{\polar{(\polar{N})}} = \ch{N}$, with $x \in \fh{\HsSet}$
and $y \in \ch{N}$.  Let $\Hs$ be the halfspace that contains 
$\fh{\HsSet}$ supported by the normal hyperplane of $xy$ through 
$y$. Let $V_\Hs$ be the conflict set of $P$ and $\Hs$.
If $V_\Hs = \emptyset$, then report that 
$\ch{P}$ and $\fh{\HsSet}$ are disjoint, and output $x$, $y$ 
as the witness. Otherwise, add to $N$ all elements of 
$V_H$ and continue with the next iteration.

If the loop terminates without a result, the algorithm finishes and returns an \textsc{Error}.

\subsection{Running Time}

While at this point we have no idea why $\Test(P, \HsSet)$
works, we can start by analyzing its running time. 
In the base case, when both $P$ and $\HsSet$ have of at most 
$\alpha$ elements, we can compute $\ch{P}$ explicitly to obtain the 
$O\Big(\alpha^{\lfloor d/2\rfloor}\Big)$ halfspaces of 
$\HsSet_P$. We can do this in a brute force manner by trying all 
$d$-tuples of $P^d$ and checking whether all of $P$ is 
on the same side of the hyperplane spanned by a given tuple, 
or we can run a convex-hull algorithm~\cite{ClarksonSh89,Chazelle93}. 
The former approach has a running time $O\big(\alpha^{d+1}\big)$, 
while the latter needs $O\big(\alpha^{\lfloor d/2\rfloor}\big)$ 
time~\cite{ClarksonSh89,Chazelle93}.
Once we have $\HsSet_P$ at hand, we can run standard LP-type algorithms 
with $O(\alpha^{\lfloor d/2\rfloor})$ constraints to determine the 
closest pair either between $\ch{P}$ and $\fh{\HsSet}$, or between 
$\ch{\polar{\HsSet}}$ and $\fh{\polar{P}}$. 
The running time of such algorithms is 
$O(d^{O(d)} \alpha^{\lfloor d/2\rfloor}) = O(d^{O(d)})$~\cite{chan2018improved}.

To see what happens in the main loop of the algorithm, we apply the
theory of $\eps$-nets, as described in the Section~\ref{sec:tools}.
As mentioned there, the initial set $N$ is a $(1/d^4)$-net for $P$.
Thus, the size of each $V_\Hs$ added to $N$ in Case~2 of the main loop 
of our algorithm is at most $n/d^4$. 
Using Lemma~\ref{lem:dualNet}, the same holds for any set
$V_y$ added in Case~1.
Thus, regardless of the case, the size of $N$ 
at the beginning of the $i$-th loop iteration is at most 
$\max\{in / d^4, \alpha\}$.

The main loop runs for at most $2d + 1$ iterations.
Thus, the size of $N$ never exceeds 
$(2d + 1)n/d^4 \leq \beta n/d^3$, for some constant $\beta > 0$.
Since the algorithm to compute the $(1/d^4)$-net $N$ for $P$ runs in time $O(d^{O(d)}n)$~\cite{chan2018improved,Chazelle01},
we obtain the following recurrence for the running time:
\[ 
T(n, m) \leq \begin{cases} 
\text{LP}\Big(\alpha + \alpha^{\lfloor d/2\rfloor}, d\Big) + 
O\big(\alpha^{\lfloor d/2\rfloor}\big), &\mbox{if } n, m \leq \alpha, \\ 
(2d+1)\cdot T\Big(m, \max\big\{\beta n/d^3, \alpha\big\}\Big) + O(d^{O(d)}n), & 
\mbox{otherwise.}
\end{cases}
\]
We look further into 
$T\big(m, \max\big\{\beta n/d^3, \alpha\}\big)$ and notice that if 
we do not reach the base case, then unfolding the recursion by one more
step yields
\[
T\big(m, \max\{ \beta n / d^3, \alpha\}\big) \leq 
(2d + 1) \cdot T\big(\max\{\beta n / d^3, \alpha\}, \max\{\beta m / d^3, 
\alpha\}\big) + O(d^{O(d)}m).
\]
Thus, by contracting two steps into one,
we get the following more symmetric relation:
\[
T(n, m) \leq (2d+1)^2 \cdot 
T\big(\max\{\beta n / d^3, \alpha\}, \max\{\beta m /d^3, \alpha\}\big) + 
O\big(d^{O(d)}(n+m) \big),
\]
for sufficiently large $n$ and $m$. Together with the base case, 
one can show by induction that  this yields a running time of $O(d^{O(d)} (n+m) )$.
\\

\textbf{Remark.} 
Because the best deterministic algorithm know for LP-type problems with $n$ constraints runs also in time $O(d^{O(d)} n)$~\cite{chan2018improved},
substantial improvements on the running time of our problem seem out of reach. 
If we allow randomization however, then we can improve in two places. 
First of all, by randomly sampling $O(\alpha \log n)$ elements of $P$, we obtain a $(1/d^4)$-net of $P$ with high probability.
Secondly, the base case could be solved with faster algorithms.
The best known randomized algorithms for LP-type problems with $n$ constraints have a running time of $O(d^2 n + 2^{O(\sqrt{d \log d})})$, which substantially improves the dependency on $d$. 
Alternatively, we could use methods to solve convex quadratic programs in the base case to find the closest pair between two H-polytopes.

\subsection{Correctness}\label{sec:Correctnes}

We show that $\Test(P, \HsSet)$ indeed tests 
whether $\ch{P}$ and $\fh{\HsSet}$ intersect.
First, we  verify that the input to
each recursive call $\Test(\cdot, \cdot)$
satisfies either condition (1) or (2). 

\begin{lemma}\label{lem:input_inv}
Let $P \subset \R^d$ be a finite point set and
$\HsSet$ a finite set of halfspaces in $d$ dimensions,
such that either \textup{(1)} $P$ is avoiding while $\HsSet$ is embracing, or 
\textup{(2)}~$P$~is embracing while $\HsSet$ is avoiding.
Consider a call of $\Test(P, \HsSet)$.
Then, the input to each recursive call satisfies
either condition \textup{(1)} or condition \textup{(2)}.
\end{lemma}

\begin{proof}
We do induction on the recursion depth.
The base case holds by assumption. 
For the inductive step, we note that if the input $(P, \HsSet)$ 
satisfies condition~(1), then $(N, \HsSet)$ also satisfies 
condition (1), for any subset $N \subseteq P$.
For condition~(2), first note that our
algorithm ensures that if $\zero \in \ch{P}$, then also $\zero \in \ch{N}$. 
This implies that if $(P, \HsSet)$ satisfies condition~(2), then 
$(N, \HsSet)$ also satisfies condition~(2). 
Finally, if $(P, \HsSet)$ satisfies condition (1), 
then by Corollary~\ref{corollary: duality} $\big(\polar{\HsSet}, \polar{P}\big)$ 
satisfies condition~(2), and vice-versa. The claim follows.
\end{proof}

We are now ready for the correctness proof.
We show that $\Test(P, \HsSet)$, with $P$ and $\HsSet$ 
satisfying either (1) or (2), computes either the closest pair 
between $\ch{P}$ and $\fh{\HsSet}$, if they are disjoint, or 
the closest pair between $\ch{\polar{\HsSet}}$ and $\fh{\polar{P}}$,
if the polytopes intersect.

We use induction on $\max\{|P|, |\HsSet|\}$. For the base case, 
when the maximum is at most $\alpha$,
our algorithm uses the brute-force method. This certainly 
provides a correct answer, by our assumptions on $P$ and $\HsSet$
and by Theorem~\ref{thm:polarPolyhedra}.

For the inductive set, we may assume that each recursive
call to $\Test(\cdot, \cdot)$ provides a correct
answer. It remains to show  (i) that the main loop
succeeds in at most $2d + 1$ iterations; and (ii) if the
main loop succeeds, the algorithm returns a valid
closest pair.

\paragraph{Number of Iterations.}
We show that the algorithm will never return an \textsc{Error},
i.e., that the loop will succeed in at most $2d + 1$ iterations.
To start, we observe that the cases in the algorithm cannot alternate:
first, we encounter only Case~2, then, we encounter only 
Case~1.
\begin{lemma}\label{lem:case_alternation}
It the main loop in  algorithm $\Test(P, \HsSet)$ encounters Case~1,
it will never again encounter Case~2.
\end{lemma}

\begin{proof}
In each unsuccessful iteration, the set 
$N$ grows by at least one element, so the convex polytope
$\fh{\polar{N}}$ becomes smaller. Once $\fh{\polar{N}}$ and
$\ch{\polar{\HsSet}}$ are disjoint, they will remain disjoint for the rest
of the algorithm, and by our inductive hypothesis, this will
be reported correctly by the recursive calls to $\Test(\cdot, \cdot)$.
\end{proof}
We now bound the number of iterations in Case~2.

\begin{lemma}\label{lem:roundsCase2}
The algorithm can have at most $d + 1$ iterations in Case~2. 
If there are $d + 1$ iterations in Case~2, then the last
iteration is successful and the algorithm terminates.
\end{lemma}

\begin{proof}
Suppose there are at least $d + 2$ iterations in Case~2.
By Lemma~\ref{lem:case_alternation}, each iteration until this point encounters Case~2.
Let $N_1 \subset N_2 \subset \dots \subset N_{d + 2}$ be the set $N$ at 
the beginning of the first $d + 2$ iterations in 
Case~2.
By Lemma~\ref{lem:cp_violator_primal} and the inductive
hypothesis, each time 
we run into Case~2 unsuccessfully, the distance between 
$\ch{N}$ and $\fh{\HsSet}$ decreases strictly. Since the first
$d + 1$ iterations in Case~2 are not successful, this means 
$d(\ch{N_i}, \fh{\HsSet}) > d(\ch{N_{i + 1}}, \fh{\HsSet})$, 
for $i = 1, \dots, d + 1$.  

Because the $(d + 2)$-th iteration runs into Case~2, it follows that 
$\ch{N_{d + 2}}$ does not intersect $\fh{\HsSet}$. Let $x$, $y$ be the closest 
pair between $\fh{\HsSet}$ and $\ch{N_{d + 2}}$, with
$y \in \ch{N_{d + 2}}$. Then, $y$ must lie on a face of
$\ch{N_{d + 2}}$, and by Carath\'eodory's 
theorem~\cite[Theorem~1.2.3]{Matousek02}, 
there is a set $B \subseteq N_{d + 2}$ with at most $d$ elements 
such that $y \in \ch{B_{d + 2}}$. 
We claim that 
in each prior iteration $i = 1, \dots, d + 1$,
the conflict set $V_H$ must contain at least one new element of 
$B$. Otherwise, if all the elements of $B$ were already in 
some $N_i$, with $i \leq d + 1$, then $\ch{N_i}$ would contain $y$ and hence 
have a distance to $\fh{\HsSet}$ smaller or equal than $\ch{N_{d + 2}}$, 
leading to a contradiction. Similarly, if in an iteration $i \leq d + 1$, all
elements of $B$ were contained in  $\Hs$, then
the distance between $\ch{N_i}$ and $\fh{\HsSet}$ could not strictly
decrease, by Lemma~\ref{lem:cp_violator_primal}.
However, $B$ contains only $d$ elements, and
we have $d + 1$ iterations, so we cannot add a new element of $B$ at the
end of each iteration
Thus, the main loop can have at most 
$d$ unsuccessful iterations in Case~2 before either encountering
Case~2 successfully or reaching Case~1.
\end{proof}

\begin{lemma}\label{lem:roundsCase1}
The algorithm can have at most $d + 1$ iterations in Case~1.
\end{lemma}

\begin{proof}
Similar to the proof of Lemma~\ref{lem:roundsCase2}, assume that
we have at least $d + 2$ iterations in Case~1.
Let 
$N_1 \subset N_2 \subset \dots \subset N_{d + 2}$ be the set $N$ 
at the beginning of each such iteration.
By Lemma~\ref{lem:cp_violator_polar} and the inductive hypothesis,
each time we encounter Case~1 unsuccessfully, 
we strictly increase the distance between 
$\ch{\polar{\HsSet}}$ and $\fh{\polar{N}}$. That is, 
$d(\ch{\polar{\HsSet}}, \fh{\polar{N_i}}) > 
d(\ch{\polar{\HsSet}}, \fh{\polar{N_{i + 1}}})$, for $i = 1, \dots, d + 1$.

Because we run into Case~1 in the $(d + 2)$-th iteration, it 
follows that $\ch{\polar{\HsSet}}$ and 
$\fh{\polar{N_{d + 2}}}$ do not intersect. Let $x$, $y$ be the closest pair 
between $\ch{\polar{\HsSet}}$ and $\fh{\polar{N_{d + 2}}}$, with
$y \in \fh{\polar{N_{d + 2}}}$. Let $B$ be the at most $d$ elements in 
$N_{d + 2}$ such that $x$, $y$ is the closest pair of $\ch{\polar{\HsSet}}$ 
and $\fh{\polar{B}}$. Note that $y$ could either be a vertex of 
$\fh{\polar{B}}$, or lie in the relative interior of  one of its faces. 
Observe that in each unsuccessful iteration 
in Case~1, $V_y$ must include a new member of $B$. 
Otherwise, if all the elements of $B$ were already in some $N_i$ with 
$i \leq d + 1$, 
then $\fh{\polar{N_i}}$ would have a distance to 
$\ch{\polar{\HsSet}}$ larger or equal 
than $\fh{\polar{N_{d + 2}}}$, leading to a contradiction. 
Similarly, if in an iteration $i \leq d + 1$, all
elements of $\polar{B}$ were not in conflict with  $y$, then
the distance between $\fh{\polar{N_i}}$ and $\ch{\polar{\HsSet}}$ 
could not strictly
decrease, by Lemma~\ref{lem:cp_violator_polar}.
However, this is impossible, because $B$ has $d$ elements and
we have at least $d + 1$ unsuccessful iterations.
Thus, in the $(d + 1)$-th iteration at the latest, 
we would observe that $V_y$ is empty and 
the algorithm would finish.  That is, the main loop can run for 
at most  $d$ unsuccessful iterations in Case~1.
\end{proof}

Lemmas~\ref{lem:case_alternation}, \ref{lem:roundsCase2}, and
\ref{lem:roundsCase1} guarantee that 
the algorithm will finish successfully within 
$2d + 1$ iterations and that it will never return an \textsc{Error}.
It remains to argue that the algorithm reports a correct closest
pair if one of the two cases is encountered successfully.

\paragraph{Correctness of the Closest Pair.}
We first analyze the success condition of Case~1, i.e., when 
$\ch{\polar{\HsSet}}$ and $\fh{\polar{N}}$ are disjoint.
This condition is triggered when we have a set 
$N \subseteq P$ and a point $y \in \fh{\polar{N}}$ such that 
$V_y = \emptyset$. 
By Lemma~\ref{lem:cp_violator_polar}, the closest pair $x$, $y$ between 
$\ch{\polar{\HsSet}}$ and $\fh{\polar{N}}$ then coincides with the closest 
pair between $\ch{\polar{\HsSet}}$ and $\fh{\polar{P}}$.
In particular, this implies that $\ch{\polar{\HsSet}}$ and $\fh{\polar{P}}$ 
are disjoint.
Because the recursive call returns correctly the closest pair $x$, $y$ 
between $\ch{\polar{\HsSet}}$ and $\fh{\polar{N}}$ by induction, 
it follows that the algorithm correctly returns the 
closest pair between $\ch{\polar{\HsSet}}$ and $\fh{\polar{P}}$.

Next, we analyze the success condition of Case~2, i.e., 
when $\ch{\polar{\HsSet}}$ and $\fh{\polar{N}}$ intersect.
This implies by Theorem~\ref{thm:polarPolyhedra} that $\ch{N}$ and 
$\fh{\HsSet}$ are disjoint. Let $x$, $y$ be the closest pair between 
$\ch{N}$ and $\fh{\HsSet}$, with $y \in \ch{N}$.
The success condition of Case~2 is triggered when $V_H = \emptyset$. 
By Lemma~\ref{lem:cp_violator_primal}, this means that
$x$, $y$ coincides with the closest pair between $\ch{P}$ and $\fh{\HsSet}$.
In particular, $\ch{P}$ and $\fh{\HsSet}$ are disjoint.
Because the recursive call returns correctly the closest pair $x$, $y$ 
between $\ch{N}$ and $\fh{\HsSet}$ by induction, 
the algorithm correctly returns the closest pair 
between $\ch{P}$ and $\fh{\HsSet}$. 
This now shows that $\Test(P, \HsSet)$ is indeed correct.

\subsection{The Final Algorithm}

Finally, we show how to remove the initial assumption 
that $(P, \HsSet)$ satisfies either condition (1) or condition (2).

\begin{theorem}
Let $P$ be a set of $n$ points in $\R^d$ and let $\HsSet$ 
be a set of $m$ halfspaces in $\R^d$. 
We can test if $\ch{P}$ and $\fh{\HsSet}$ intersect in 
$O(d^{O(d)} (n + m))$ time. If they do, then we compute a point in their 
intersection; otherwise, we compute a separating plane.
\end{theorem}
\begin{proof}
Recall that our algorithm requires that either 
(1) $\zero\notin \ch{P}$ and $\zero \in \fh{\HsSet}$, or 
(2) $\zero \in \ch{P}$ and $\zero \notin \bigcup_{\Hs \in \HsSet} \Hs$ 
to work, which might not hold for the given $P$ and $\HsSet$.
Thus, before running $\Test(P, \HsSet)$, we first compute a 
point in $\fh{\HsSet}$ using standard linear programming 
and change the coordinate system so that this point coincides with $\zero$. 
Then, we test using standard linear programming if $\zero\in \ch{P}$. 
If it is, we are done.
Otherwise, we guarantee that condition (1) is satisfied, and 
we can run $\Test(P, \HsSet)$ in $O(d^{O(d)} (n + m))$ time.
\end{proof}

\paragraph{Acknowledgments.}
This work was initiated at the \emph{Sixth Annual Workshop on Geometry 
and Graphs}, that took place March 11--16, 2018, at the
Bellairs Research Institute. We would like to thank the organizers
and all participants of the workshop for stimulating discussions
and for creating a conducive research environment. We would also like
Timothy M.~Chan for answering our questions about LP-type problems
and for pointing us to several helpful references.

\bibliographystyle{abbrv}
\bibliography{acit}

\end{document}